\RequirePackage{fix-cm}

\documentclass[smallcondensed]{svjour3}     
\journalname{Bulletin of Mathematical Biology}
\usepackage{mathptmx}
\usepackage{setspace}
\doublespace
\smartqed  
\usepackage[sort&compress]{natbib}
\usepackage{graphicx}

\begin{document}

\title{Variable Susceptibility With An Open Population}
\subtitle{A Transport Equation Approach}
\titlerunning{Transport Equations and Open Populations}
\author{Benjamin R. Morin}
\institute{Benjamin R. Morin
\at ecoSERVICES Group\\
School of Life Sciences, Arizona State University\\
PO Box 874501, Tempe, Arizona 85287-4501, USA\\\email{brmorin@asu.edu}}
\maketitle
\abstract{Variable individual response to epidemics may be found within many contexts in the study of infectious diseases (e.g., age structure or contact networks). There are situations where the variability, in terms of epidemiological parameter, cannot be neatly packaged along with other demographics of the population like spatial location or life stage. Transport equations are a novel method for handling this variability via a distributed parameter; where particular parameter values are possessed by various proportions of the population. Several authors (e.g., Kareva, Novozhilov, and Katriel) have studied such systems in a closed population setting (no births/immigrations or deaths/emigrations), but have cited restrictions to employing such methods when entry and removal of individuals is added to the population.  This paper details, in the context of a simple susceptible-infectious-recovered (SIR) epidemic, how the method works in the closed population setting and gives conditions for initial, transient, and asymptotic results to be equivalent with the nondistributed case. Additionally, I show how the method may be applied to various forms of open SIR systems.  Transport equations are used to transform an infinite dimensional system for the open population case into a finite dimensional system which is, at the very least, able to be numerically studied, a model with direct inheritance of the distributed parameter is shown to be qualitatively identical to the nondistributed case, and finally a model where disease results in sterilization is fully analyzed.}
\section{Introduction}
Epidemic compartmental models typically consider a  population broken up into a number of compartments that describe the individual disease state of its members. Flow between these compartments is driven by a set of biologically motivated parameters commonly assumed to be the average quantity for a given population (e.g., the average duration of infection).  Within this paper it is assumed that the population is not homogenous in its parameters (i.e., individuals may be stratified by a differential response to disease in some way). This has certainly been entertained in many contexts such as spatial heterogeneity (e.g., networks and traveling wave solutions) \citep{gertsbakh1977epidemic, pastor2001epidemic,moreno2002epidemic,newman2006structure,house-et-al-09,joo2004pair,hoppensteadt1975mathematical,castillo1998global,busenberg1988endemic}, longitudinal behavior heterogeneity \citep{herrera2011multiple,fenichel2011adaptive}, and, most relevantly here, differential disease response due to age \citep{hoppensteadt1975mathematical,castillo1998global,busenberg1988endemic}. 

Work by Karev \citep{karev2005dynamic,karev2005dynamicsb}, utilized by Novozhilov \citep{novozhilov2008spread} and Katriel \citep{katriel2012size}, considered such heterogeneities for models but approached the characterization of the dynamics from a different mathematical construction.  Rather than a proliferation of compartments or a partial differential formulation they utilize transport equations/variables, specifically within the context of epidemic spread within a closed population. Phenomenologically, the heterogeneity of disease response has had general theory applied to it via the use of transport equations/variables for models of the form
\begin{eqnarray}
\nonumber\dot{X}(t,w)&=&X(t,w)F(X(t,w;\vec{\theta}),\vec{Y(t)},t;\vec{\theta}),\\
\vec{\dot{Y}(t)}&=&G(X(t,w),\vec{Y(t)},t;\vec{\theta}).\label{general}
\end{eqnarray}
 System \ref{general} includes a state variable, $X$, which depends on a parameter $w$. This parameter is unique from the other system parameters, contained in $\vec{\theta}$, in that it is distributed within the population. In other words $w$ is a random variable with $\frac{X(t,w)}{\int X(t,w)dw}$ being the proportion of the type-$X$ population possessing the parameter value $w$, the time dependent probability distribution if you will. This formulation is most useful when the parameter value for a given individual is independent of time; thus, given the current state of the theory, susceptibility to a disease that evolves as one ages would be inappropriate to model given this technique. Indeed, the application of the transport equations to epidemics have often been in a closed-population $SIR$  setting.  Within an $SIR$ setting individuals are either susceptible to a disease, $S$, infected/infectious, $I$, or have recovered and are now immune to reinfection, $R$.  Defining the rate that a contact between a susceptible and an infectious individual results in a new infection, $\beta$, and an individual recovery rate from the disease, $\gamma$, while additionally assuming contacts are made at random within a population results in the system
\begin{eqnarray*}
\dot{S}(t)&=&-\beta S(t)\frac{I(t)}{N},\\
\dot{I}(t)&=&\beta S(t)\frac{I(t)}{N}-\gamma I(t),\\
\dot{R}(t)&=&\gamma I(t),
\end{eqnarray*}
where $\dot{X}(t)=\frac{dX(t)}{dt}$. The inclusion, in Section \ref{ClosedSIR}, of the partially replicated results of introducing variable susceptibility into an $SIR$ model as above is meant to explain the use of the theory used to study equation \ref{general}, and includes brief discussion of {\em why it works}. Subsection \ref{equivalence} demonstrates a different technique for proving the equivalence of asymptotic behavior for distributed models of the form given in System \ref{general} and their undistributed counter parts\footnote{This dynamic equivalence is similar to that of the age structured model, see \citep{Brauer01}.}. It is critical to note here, and reiterate throughout the paper, what is ment by dynamic equivalence. 

It is easy to verify that the rate of change of the distributed variable, $w$, is equal to its variance within the population at each moment in time (a version of Fisher's fundamental theorem \citep{karev2010mathematical}).  Therefore, as the parametrically heterogeneous epidemic system evolves with time, the basic statistics (e.g., expected value and variance) of $w$ will also change. As in \citep{kareva2012preventing, kareva2011niche}, and references contained therein, the distributed variable is first considered to take on a single value (i.e., as a random variable it has the delta distribution) and one constructs a bifurcation diagram for the now parametrically homogeneous system (i.e., the standard undistributed version of the differential equations). The bifurcation parameter, $\rho(w)$, is chosen to be a function of the (now homogeneous) distributed parameter. As the statistics of $w$ evolve with the parametrically heterogeneous system, the current time value of $\rho(w)$ ``travels" through the parametrically homogeneous phase-parametric portrait \citep{karev2006mathematical}. Therefore, the {\em dynamic equivalence}, between distributed and unditributed systems (parametrically hetero- and homo- geneous respectively) is ment to imply that the bifurcation diagram created from the undistributed case characterizes the nature of the fixed points of the distributed system and may be used to qualitatively discuss the transient behavior of the latter.

The representation theory explained in \citep{karev2005dynamic,karev2005dynamicsb} excludes models that exhibit {\em blue-sky births} (i.e., entries into the distributed class at a rate not proportional to the class itself), and such inclusions have been avoided by both Novozhilov and Katriel.  In the beginning of Section \ref{sec:BlueSky} it is demonstrated that a simple inclusion of births and deaths into the $SIR$ model results in a system where the transport system theory is not useful to its fullest, but {\em importantly} constructs an equivalent system that may be solved numerically without approximation outside of numerical methods.  A model where disease-state and susceptibility is passed onto offspring and a second where the disease results in permanent sterilization of the individual are discussed in Subsections \ref{subsec:Inher} and \ref{subsec:Steri}.  Both are shown to be fully applicable to the established theory, using similar analysis techniques as in \citep{kareva2012transitional}, with the former reducing to the study of the closed population model and the latter presenting novel dynamics.

\section{Differential Susceptibility $SIR$ With a Closed Population}\label{ClosedSIR}
Introduce susceptible variability\footnote{The full derivation of this model may be found in \citep{novozhilov2008spread}.} via a parameter $w$ and the resultant value of $\beta(w)$.  Assume that for all $w$, $0\leq\beta(w)<\infty$ (to ensure all populations remain positive and finite in finite time), and denote the susceptible individuals with the particular susceptibility of $w$ via $S(t,w)$. The resulting system takes on the form
\begin{eqnarray}
\nonumber\dot{S(t,w)}&=&-\beta(w)S(t,w)\frac{I(t)}{N},\\
\dot{I(t)}&=&\int\beta(w)S(t,w)dw\frac{I(t)}{N}-\gamma I(t),
\label{DistSIR}
\end{eqnarray}
with $R(t)$ omitted due to the constant population size. Susceptibility has no impact on an individual once infected and thus $I(t):=\int I(t,w)dw$, the total count of all infected individuals, is utilized within the incidence term. Based on the representation theory of Karev \citep{karev2005dynamic,karev2005dynamicsb,karev2010replicator}, one introduces a transport variable $\dot{q}(t)=-\frac{I(t)}{N}$ and through separation of variables for the susceptible class finds 
$$S(t,w)=S(0,w)e^{\beta(w)q(t)}.$$ 
Subsequently one may note that $S(t)=\int S(t,w)dw$ is the total susceptible population and that $S(t)$ satisfies
$$\dot{S}(t)=-\overline{\beta(t)}S(t)\frac{I(t)}{N},$$
where 
\begin{eqnarray*}
\overline{\beta(t)}&=\frac{\int\beta(w)S(t,w)dw}{\int S(t,w)dw}
&=\frac{d}{d\lambda}\left[\ln\left(M_\beta(0,\lambda)\right)\right|_{\lambda=q(t)}.
\end{eqnarray*}
$M_\beta(t,\lambda)$ is the moment-generating function of the time $t$ density of property $w$ within the susceptible population.  The system in Equation \ref{DistSIR} may be of arbitrarily large dimension, since $w$ may take on values along a continuum, and is now reduced to two non-autonomous differential equations and an integral expression\footnote{Interestingly, this method constructs non-autonomous differential equations which implies that one is still searching an infinite dimnesional solution space.  This has shifted the continuum of state variables onto a transport variable ODE and a time dependent parameter. However, with the introduction of this transport variable, coupled with the ability to solve $S(t,w)$ in terms of it and initial data, the integral expressions are ``solvable" numerically and may be represented via moment-generating functions of the initial data.}, or a {\em transport system} given by:
\begin{eqnarray}
\nonumber\dot{S}(t)&=&-\overline{\beta(t)}S(t)\frac{I(t)}{N},\\
\dot{I}(t)&=&\overline{\beta(t)}S(t)\frac{I(t)}{N}-\gamma I(t),\label{SystemSIR}\\
\nonumber\dot{q}(t)&=&-\frac{I(t)}{N},\\
\nonumber\overline{\beta(t)}&=&\frac{\int\beta(w)S(0,w)e^{\beta(w) q(t)}dw}{\int S(0,w)e^{\beta(w)q(t)}dw}.
\end{eqnarray}
It may be assumed that the initial condition for the distribution of $S(0,w)$ and the values for $\beta(w)$ are known.  

The system \ref{SystemSIR} may be recast where the incidence is a nonlinear function (useful for the calculation of final epidemic size relation). Denote the moment-generating function for the distribution describing the selection of an arbitrary susceptible individual with susceptibility $\beta(w)$ via
$$M_\beta(0,q(t))=\int \frac{S(0,w)}{S(0)}e^{\beta(w)q(t)}dw,$$
and rewrite $\dot{S}(t)$ by first noting
\begin{eqnarray*}
\frac{1}{S(t)}\frac{d}{dt}S(t)&=&
\frac{d}{d\lambda}\left[\ln\left(M_\beta(0,\lambda)\right)\right|_{\lambda=q(t)}\frac{d}{dt}q(t),\\
\frac{d\ln(S(t))}{dt}&=&\frac{d}{dt}\ln(M_\beta(0,q(t))),\\
\frac{S(t)}{S(0)}&=&M_\beta(0,q(t)),\\
q(t)&=&M_\beta^{-1}(0,S(t)/S(0)),
\end{eqnarray*}
to get
$$
\dot{S}(t)=-\frac{d}{d\lambda}M_\beta(0,\lambda)|_{\lambda=M_\beta^{-1}(0,S(t)/S(0))}S(0)\frac{I(t)}{N}.
$$
By the inverse function theorem\footnote{$(f^{-1})'(b)=\frac{1}{f'(a)}$ where $b=f(a)$.} this results in
$$
\dot{S}(t)=-\left[\frac{d}{d\lambda}M_\beta^{-1}(0,\lambda)|_{\lambda=S(t)/S(0)}\right]^{-1}S(0)\frac{I(t)}{N}=-h(S(t))\frac{I(t)}{N},
$$
where $h(S(t))$ is a non-linear function of $S(t)$.  Calculating $\frac{dS(t)}{dR(t)}$ results in the expression the final epidemic size must satisfy:
$$
S_\infty=S(0)M_\beta\left(0,\frac{S_\infty-N}{N\gamma}\right).
$$
Furthermore, as a straightforward application of the results in \citep{Blythe91}, $SIR$ models with this nonlinear form have a basic reproduction number\footnote{The limitations of such a quantity should be apparent in such a case where the infectivity is a function of time. Nevertheless, it is presented as a standard threshold computation.} $R_0=\frac{\overline{\beta(0)}}{\gamma}$.  With these two threshold quantities one may address when the distributed and non-distributed cases are identical (either initially, asymptotically, or both).

Equating the two basic reproduction numbers results in $\beta=\overline{\beta(0)}.$ Thus, if the traditional $\beta$ is chosen to be the initial mean of the distribution of the resultant infection rate, then the initial behavior of the two models is identical.  Supposing that the solution to $S_\infty=S(0)e^{-\frac{\beta}{\gamma}\left(1-\frac{S_\infty}{N}\right)}$ is identical to that of the distributed problem implies this solution must satisfy
$$e^{\frac{\beta}{\gamma}\frac{R_\infty}{N}}=M_\beta\left(0,\frac{1}{\gamma}\frac{R_\infty}{N}\right),$$
where $R_\infty$ is the limiting recovered population, identical on each side of the expression. Note that the $\beta$ on the left hand side is the particular  value (from the classical model) while that on the right is a distributed variable.  Since distributions are uniquely identified by their Moment-generating function, we may conclude that for identical final epidemic sizes the initial distribution of $w$ in the susceptible population must be delta.
Thus if the final size is identical between the distributed and non-distributed cases, then the initial epidemic behavior is identical; however, the converse is not true.  One may choose any number of initial distributions such that the mean at time zero is equivalent to $\beta$. This is particularly important when estimating parameters from initial epidemic data (the initial phase of exponential growth).  These estimations typically assume a delta distribution of infectivity and therefore may be used to incorrectly project final epidemic size which will {\em always be an over estimation of spread when compared to the distributed case}.  

\subsection{Dynamic Equivalence}\label{equivalence}

For the undistributed model, all points of the form $(S^*,0)$ are equilibria.  Qualitatively, this implies that fixed points where $S^*>\frac{\gamma}{\beta}N$ are unstable and where $S^*<\frac{\gamma}{\beta}N$ are stable. The fixed points for the transport system \ref{SystemSIR} may pose a particular challenge because the system is now non-autonomous. However, assuming a non-degenerate situation (i.e., $\overline{\beta(t)}\neq0$), the equilibria are still of the form $(S^*,0)$. The linearization of the distributed system gives the condition for stability as
$$S(t)<\frac{\gamma}{\overline{\beta(t)}}N.$$ 
One may show that the stability threshold may not create a complicated phase space where $\overline{\beta(t)}$ forms an implicit (in time) boundary which may induce oscillations (necessarily damped) in the phase space due to the monotonicity of $\overline{\beta(t)}$.  Define the threshold $T(t)=\frac{\gamma N}{\overline{\beta(t)}}$ and consider 
\begin{eqnarray}
\nonumber\frac{dT(t)}{dt}&=&-\frac{\gamma N}{\left(\overline{\beta(t)}\right)^2}\frac{d\overline{\beta(t)}}{dt},\\
\nonumber&=&\frac{\gamma Var(\beta(t))I(t)}{\left(\overline{\beta(t)}\right)^2}.
\label{Threshold}
\end{eqnarray}
By Equation \ref{Threshold} it is clear that $T(t)$ is monotonically increasing (furthermore, its slope approaches $0$  as $Var(\beta(t))$ approaches 0, i.e., as $\frac{S(t,w)}{S(t)}$ approaches a singular distribution).  Since $T(t)$ is monotonically increasing, the amount of the $S(t)$-axis in the phase space for which the points are stable is also increasing (non-decreasing in the event that $\frac{\gamma}{\overline{\beta(t)}}>1$ for some $t<\infty$). Since $I(t)$ begins to decrease once it crosses $T(t)$, and due to the monotonicity of $T(t)$, there is no way to induce an oscillation on $I(t)$ regardless of the distribution on $w$ (i.e., once $I(t)$ decreases it may not again increase).  Similarly, there will be no oscillations in either $S(t)$ or $S(t,w)$.  This is equivalent to arguing that $\overline{\beta(t)}$ is monotonically decreasing in time. Therefore, as it travels through the bifurcation diagram of the undistributed case it will only pass the homogeneous system's bifurcation point of $\frac{\gamma}{\beta}N$ once. I posit a lemma as the conclusion of this section:

\begin{lemma}[Closed SIR Equivalence/Worst Case Distribution] Due to the monotonicity of $\overline{\beta(t)}$ ,the transient and asymptotic qualitative behaviors of the distributed and non-distributed $SIR$ models are identical.  Additionally, the initial behavior and final epidemic size of the two models are identical if $S(0,w)=\delta_{\beta(w)-\beta}S(0)$; the initial behavior of the two models are identical if and only if $\overline{\beta(0)}=\beta$. Also due to the monotonic decrease in $\overline{\beta(t)}$, over all distributions chosen with equivalent initial mean, the most infection is produced by the delta distribution (non-distributed model).\label{lemma: distributed model}
\end{lemma}
\begin{proof}
The equivalence claims are all proven in the text preceding this Lemma.  To prove the worst case scenario claim observe that since $q(0)=0$, and $I(t)\rightarrow0$, we have that $q(t)$ monotonically decreases to some value $\eta\in(-\infty,0)$. The derivative of $\overline{\beta(t)}$ with respect to $q(t)$ is $Var(\beta(t))>0$, implying $\overline{\beta(t)}$ decreases monotonically to $\epsilon\in(0,\overline{\beta(0)})$. Note the final epidemic size calculation
$$
S_\infty=S(0)e^{-\int_0^\infty\overline{\beta(u)}I(u)du},
$$
and the inequality
$$
\int_0^\infty\overline{\beta(u)}I(u)du\leq\overline{\beta(0)}\int_0^\infty I(u)du=-\frac{\overline{\beta(0)}}{\gamma}\frac{(S_\infty-N)}{N}.
$$
This implies that the final size of the susceptible population, $S_\infty$, for the undistributed case is minimal given that the distribution has an equivalent initial mean susceptibility, i.e., 
$$
S(0)e^{-\int_0^\infty\overline{\beta(u)}I(u)du}\geq S(0)e^{-\frac{\overline{\beta(0)}}{N\gamma}(N-S_\infty)}.
$$
\end{proof}\qed

\section{``Blue-Sky" Births \& Open Populations}\label{sec:BlueSky}

The most straightforward manner to ``open" the population of the aforementioned $SIR$ model is to suppose newborns are susceptible and are birthed from each epidemiological class which experience proportionate removal from the system.  This results in the distributed system
\begin{eqnarray}
\nonumber\dot{S}(t,w)&=&\Lambda N(t,w)-\beta(w)S(t,w)\int\frac{I(t,w)}{N(t)}dw-\mu S(t,w),\\
\dot{I}(t,w)&=&\beta(w)S(t,w)\int\frac{I(t,w)}{N(t)}dw-(\gamma+\mu)I(t,w),\label{Wideopen}
\\
\nonumber N(t,w)&=&N(0,w)e^{(\Lambda-\mu)t}.
\end{eqnarray}
However, both $I(t,w)$ and $R(t,w)$ produce members of $S(t,w)$ (and thus the ODE for $S(t,w)$ cannot be solved via separation of variables).  Nevertheless, continuing as before, let $\dot{q}(t)=-\frac{I(t)}{N(t)}$ to get
$$S(t,w)=\left(\Lambda N(0,w)\int_0^te^{\Lambda r-\beta(w)q(r)}dr+S(0,w)\right)e^{\beta(w)q(t)-\mu t}.
$$
The method for solving this equation was via the integrating factor $e^{-\beta(w)q(t)+\mu t}$ as opposed to separation of variables, as in the closed population case, but as before dependent on intial data and the current value of $q(t)$. The complication comes into play by requiring knowledge of the past values of $q(t)$ through the integral.  We may now consider the transport system
\begin{eqnarray}
\nonumber\dot{S}(t)&=&\Lambda N(t)-\overline{\beta(t)}S(t)\frac{I(t)}{N(t)}-\mu S(t),\\
\nonumber\dot{I}(t)&=&\overline{\beta(t)}S(t)\frac{I(t)}{N(t)}-(\gamma+\mu)I(t),\\
\nonumber N(t)&=&N(0)e^{(\Lambda-\mu)t},\\
\nonumber\overline{\beta(t)}
&=&\frac{\Lambda\int \beta(w)P(w)e^{\beta(w)q(t)}\int_0^te^{\Lambda r-\beta(w)q(r)}drdw+\int \beta(w)P_S(w)e^{\beta(w)q(t)}dw}{\int \left(\Lambda P(w)\int_0^te^{\Lambda r-\beta(w)q(r)}dr+P_S(w)\right)e^{\beta(w)q(t)}dw},\\
\nonumber\dot{q}(t)&=&-\frac{I(t)}{N(t)},
\end{eqnarray}
where $P(w)=\frac{N(0,w)}{N(0)}, P_S(w)=\frac{S(0,w)}{N(0)}\approx P(w),$ and $M_{\beta|S}$ is the moment generating function conditioned on $P_S(w)$.  It should be clear that the set of possible qualitative behaviors from the undistributed case\footnote{The undistributed system is a homogeneous system and thus, by rescaling to proportionate variables we may equate stability analysis of fixed points of the rescaled system to stability analysis of the exponential trajectories of the original, undistributed system. The disease free equilibrium (trajectory) is attracting if and only if $\beta<\gamma+\Lambda$.  When the disease free state is not attracting there is an endemic equilibrium (trajectory) in the relevant phase space which is stable.} are the only options for the evolution of this transport system.  However, the nature of $\overline{\beta(t)}$ (whether it is increasing, decreasing, or both) is left as an open problem. The difficulty is highlighted when considering the derivative of $\overline{\beta(t)}$ with respect to $t$:
$$
-\frac{I(t)}{N(t)}Var\left(\overline{\beta(t)}\right)+\frac{\Lambda\left(\overline{\beta(t)}-\overline{\beta_N(t)}\right)}{S(t)}\left(\mu N(t)-\overline{\beta_N(t)}\right),
$$
where $\overline{\beta_N(t)}=\frac{\int\beta(w)N(0,w)dw}{N(t)}$.  Note that $\overline{\beta(0)}\approx\overline{\beta_N(0)}=\overline{\beta_N(t)}$.  Thus at time $t=0$ it is true that $\overline{\beta(t)}$ is decreasing.  For $t>0$ the sign of $\frac{\Lambda\left(\overline{\beta(t)}-\overline{\beta_N(t)}\right)}{S(t)}\left(\mu N(t)-\overline{\beta_N(t)}\right)$ is equivalent to that of $$\left(\overline{\beta(t)}-\overline{\beta_N(t)}\right)\left(\mu e^{(\Lambda-\mu)t}\left(\int N(0,w)dw\right)^2-\int\beta(w)N(0,w)dw\right).$$ Given this information it is feasible that the exponential trajectory for the transport system could oscillate between being attracted to the disease free trajectory and the endemic trajectory.  

{ Furthermore, if $\overline{\beta(w)}\geq1$ for all $w$ then the derivative of $\overline{\beta(t)}$ with respect to $q(t)$ is always positive:
$$
\frac{d\overline{\beta(t)}}{dq(t)}=Var(\overline{\beta(t)})+\frac{\Lambda}{S(t)}\int N(t,w)\left(\overline{\beta(t)}-1\right)dw>0.
$$
Since $q(t)$ is monotonically decreasing we may infer in this case that there exists a time $\tau<\infty$ such that for all $t\geq\tau$, $\overline{\beta(t)}<\gamma+\Lambda$.  This implies that the disease will eventually ``burn itself out" and the disease free trajectory will be stable. }

This does not seem to have opened many analytical pathways as in the closed case, however this should be seen as a boon for numerical computation.  The original system involved (in general) an infinite number of ordinary differential equations to integrate numerically.  However, the above intergo-differential system consists of a finite number of equations to solve numerically (involving and initial conditions $S(0,w), N(0,w),$ and $\beta(w)$ and the solution trajectory of $q(t)$ up to and including the current time); a tractable computation problem which will have an identical solution to the infinite dimensional ODE case and does not involve approximation with respect to the dimensionality of the system is therefore possible. 

\subsection{Pure Inheritance}\label{subsec:Inher}
A method to circumvent the {\em blue-sky births} into $S(t,w)$ is to assume the malady, immunity to it, and the susceptibility to it is transferred to new borns.  This inheritance mechanism is weak at best because 1) the additions and removals to the system are not solely births and deaths in general but could be immigration and emmigration from the area in question and 2) we have to further assume the father's status confers nothing onto new-borns.  With these caveats in mind, one may formulate:
\begin{eqnarray}
\nonumber\dot{S}(t,w)&=&\Lambda S(t,w)-\beta(w)S(t,w)\frac{I(t)}{N(t)}-\mu S(t,w),\\
\nonumber\dot{I}(t,w)&=&\Lambda I(t,w)+\beta(w)S(t,w)\frac{I(t)}{N(t)}-(\gamma+\mu)I(t,w),\\
\nonumber\dot{R}(t,w)&=&\Lambda R(t,w)+\gamma I(t,w)-\mu R(t,w),\\
\nonumber N(t,w)&=&N(0,w)e^{(\Lambda-\mu)t}.
\end{eqnarray}
The solution to $S(t,w)$ may then be found via separation of variables as
$$S(t,w)=S(0,w)e^{(\Lambda-\mu)t+\beta(w)q(t)},$$
with $\dot{q}(t)=-\frac{I(t)}{N(t)}$. Integrating each ODE over $w$ gives the transport system
\begin{eqnarray}
\nonumber\dot{S}(t)&=&\Lambda S(t)-\overline{\beta(t)}S(t)\frac{I(t)}{N(t)}-\mu S(t),\\
\nonumber\dot{I}(t)&=&\Lambda I(t)+\overline{\beta(t)}S(t)\frac{I(t)}{N(t)}-(\gamma+\mu)I(t),\\
\dot{R}(t)&=&\Lambda R(t)+\gamma I(t)-\mu R(t), \label{opensystem}\\
\nonumber N(t)&=&N(0)e^{(\Lambda-\mu)t},\\
\nonumber \overline{\beta(t)}&=&\frac{\int\beta(w)S(0,w)e^{\beta(w)q(t)}dw}{\int S(0,w)e^{\beta(w)q(t)}dw},\\
\nonumber\dot{q}(t)&=&-\frac{I(t)}{N(t)}.
   \end{eqnarray}
We may recast System \ref{opensystem} into a system with proportionate variables $s(t)=\frac{S(t)}{N(t)}$, $i(t)=\frac{I(t)}{N(t)}$ and $r(t)=\frac{R(t)}{N(t)}$, each trapped within the interval $[0,1]$.  The resulting non-autonomous system is given by
\begin{eqnarray}
\nonumber \dot{s}(t)&=&-\overline{\beta(t)}s(t)i(t),\\
\nonumber \dot{i}(t)&=&\overline{\beta(t)}s(t)i(t)-\gamma i(t),\\
\nonumber \dot{r}(t)&=&\gamma i(t),\\
\nonumber \dot{q}(t)&=&-i(t),
\end{eqnarray}
with the definition of $\overline{\beta(t)}$ left unchanged.  This system exhibits the same dynamics as the closed $SIR$ population transport equations in System \ref{SystemSIR}, save that the actual population counts travel along exponential solution trajectories. This implies that opening the population as in System \ref{opensystem} may not induce oscillations, where in the original open system given by \ref{Wideopen} we were not able to definitively rule out oscillatory behavior (it could not be shown that $\overline{\beta(t)}$ was monotonic).  

\subsection{Sterilization}\label{subsec:Steri}
The zoonoses {\em Trichomoniasis}, {\em Salmonellosis}, and {\em Leptospirosis} are infections in cows  that may impart sterility on the individual \citep{vandeplassche1982reproductive}.  Once a heifer has been infected with these diseases the next pregnancy will result in abortion. With {\em Salmonellosis} and {\em Leptospirosis} it is unclear if future pregnancies result in abortions even if the cow shows no signs of infectiousness, but upon true recovery, after a short time spent immune to the disease, the heifer is again susceptible to infection and may conceive and calf normally until reinfected.  This dynamic, similar to an $SIS$ model (the immunity is so short that the rate from $R$ to $S$ will be disproportionately large), can be shown to be completely incompatible with the transport equation technique\footnote{I've omitted showing the calculations for $SIS$ and $SIRS$ models but the reentry into the susceptible class causes the distributed equations to be completely unsolvable in any meaningful way.  The solution for $I(t)$ in the distributed susceptibility $SIS$ model looks very similar to the solution of the non-autonomous $SIS$ model \citep{lopez2010logistic}, but it may be shown that the solution is both implicit (the parameters ``depend" on $I(t)$) and incomplete (the parameters require that $I(t,w)$ be solved, which cannot be done).}. 

{\em Papillomaviruses} in sheep have both an acute and chronic stage.  During the acute stage the sheep is infectious and any pregnancy during which the sheep is in the acute phase will result in abortion \citep{oriel1974sexually}. The passing to the chronic phase causes scarification of the fallopian tubes, as it does in humans.  This scarring causes infertility in addition to making the sheep more susceptible to other STDs and STIs. While in the chronic phase the sheep is still infectious, but at a much lower level than when in the acute phase \citep{oriel1974sexually}.  I simplify this dynamic by supposing the infections caused by sheep in the chronic phase is negligible and cast the dynamics into an $SIR$ setting with variable susceptibility.  The variable susceptibility serves an amalgamation of effects that contribute to susceptibility: nutrition, infection history, cleanliness of environment, etc....

The following model suppose a population whose growth is naturally limited, modeled via logistic growth, and is single sex (females only). I introduce papillomavirus into the population noting that it 1) causes no death due to infection and 2) causes permanent infertility in infectious (acute) and recovered/immune (chronic) individuals. Suppose a logistic growth for the population given by 
$$
\dot{N}(t)=\lambda N(t)-\frac{\lambda}{K}N^2(t),
$$
and rationalize the terms mechanistically as a birth process $\lambda N(t)$ and a density dependent death process $\lambda N(t)\frac{N(t)}{K}$.  By introducing a sterilizing disease, and imparting differential susceptibility, one arrives at
\begin{eqnarray}
\nonumber \dot{S}(t,w)&=&\lambda S(t,w)\left(1-\frac{N(t)}{K}\right)-\beta(w)S(t,w)\frac{I(t)}{N(t)},\\
\dot{I}(t)&=&\int\beta(w)S(t,w)dw\frac{I(t)}{N(t)}-\left(\gamma+\lambda\frac{N(t)}{K}\right)I(t),\label{logisticopen}\\
\nonumber \dot{R}(t)&=&\gamma I(t)-\lambda R(t)\frac{N(t)}{K},\\
\nonumber \dot{N}(t)&=&\lambda \int S(t,w)dw-\lambda N(t)\frac{N(t)}{K}.
\end{eqnarray}
Introduce the transport variables $\dot{u}(t)=-\frac{N(t)}{K}$ and $\dot{v}(t)=-\frac{I(t)}{N(t)}$ to arrive at
$$S(t,w)=S(0,w)e^{\lambda t+u(t)+\beta(w)v(t)},$$
and thus
$$S(t)=e^{\lambda t+u(t)}\int S(0,w)e^{\beta(w)v(t)}dw.$$
Defining 
$$\overline{\beta(t)}=\frac{\int\beta(w)S(0,w)e^{\beta(w)v(t)}dw}{\int S(0,w)e^{\beta(w)v(t)}dw},$$
and
%
supposing the total population is less than $K$, one may rescale to state variables in $[0,1]$ and define the biologically valid domain via $T=\{(s,i,n)|s\geq0, i\geq0,n\in[0,1],s+i\leq 1\}$: 
\begin{eqnarray}
\nonumber\dot{s}(t)&=&\lambda s(t)(1-n(t))-\overline{\beta(t)}s(t)\frac{i(t)}{n(t)},\\
\nonumber\dot{i}(t)&=&\overline{\beta(t)}s(t)\frac{i(t)}{n(t)}-(\gamma+\lambda n(t))i(t),\\
\nonumber\dot{u}(t)&=&-n(t),\\
\nonumber\dot{v}(t)&=&-\frac{i(t)}{n(t)},\\
\nonumber\dot{n}(t)&=&\lambda s(t)-\lambda n^2(t),\\
\nonumber\overline{\beta(t)}&=&\frac{\int\beta(w)s(0,w)e^{\beta(w)v(t)}dw}{\int s(0,w)e^{\beta(w)v(t)}dw}.
\end{eqnarray}

The qualitative behavior of the undistributed case forms the bifurcation diagram that this non-autonomous transport system now moves through (the bifurcation parameter is a function of time).  To determine both the bifurcation diagram and the transient dynamics of the distributed parameter value one may consider the set of {\em temporary} fixed points for the distributed system that depends on the values for $\overline{\beta(t)}$. Note here what is and is not being done. By the methods of \citep{kareva2012preventing, kareva2011niche,karev2006mathematical,karev2010mathematical} one must construct a bifurcation diagram for the undistributed case and then note that the bifurcation parameter will ``travel" through this for the distributed system (thus defining the transient behavior for the distributed system). This is not studying a general nonautonomous system by freezing time and performing typical autonomous qualitative analysis; although it does seem that way. I've found that this ``breaking of the rules" is one of the easier ways to demonstrate what the distributed parameter system is doing and forms an identical analysis to what is supposed to be done (i.e., construction of a bifurcation diagram and then consider ``travel" through it). While the following contains linearizations around ``fixed points" which depend on time this is simply a short cut. In fact both the distributed and undistributed systems are analyzed simultaneously (any ``fixed point" for the nonautonomous/heterogeneous/distributed system is a true fixed point for the autonomous/homogeneous/undistributed system).

I refer to the set of  temporary fixed points as the ``fixed curve".  This ``fixed curve" is a trajectory in ${R}^3$ and should the trajectory of the state variables, $s$, $i$, and $r$, come in contact with it, in the space-time sense, then their dynamics will cease for a moment.  However, if $i(t)\neq0$ then $\dot{v}(t)\neq0$ and $\overline{\beta(t)}$ may change; this results in a departure of the state variable trajectory from the fixed curve. The state-dynamics will then not be at equilibrium and continue to evolve.  The ``fixed curve" therefore corresponds to turning points (local minimums, maximums or inflection points) that occur for all three states simultaneously. For a given value of $\beta$ in the undistributed system the fixed point(s) for the system will lie on this fixed curve. There are two simple fixed points for the undistributed system:
$$
(s^*,i^*,r^*)=(0,0,0), (1,0,0).
$$
The trivial fixed point, all states $0$, is a saddle-type node (i.e., attracting down the $i(t)$ and $r(t)$ axes and repelling down the $s(t)$-axis).  The disease free equilibrium, DFE,  $(1,0,0)$ has eigenvalues $-\lambda, -\lambda,$ and $\overline{\beta(t)}-\lambda-\gamma$. Thus {\em when} $\overline{\beta(t)}<\lambda+\gamma$ the disease free equilibrium is a stable node, at other times it is a saddle-type node. The nullclines of $s(t)$ and $i(t)$ have an intersection at $\left(\frac{(\gamma+\lambda n(t))n(t)}{\overline{\beta(t)}},\frac{\lambda n(t)(1-n(t))}{\overline{\beta(t)}}\right)$.  One may then find the $r$-coordinate via $\dot{r}(t)=\gamma i(t)-\lambda r(t)n(t)=0$ which implies 
$$
r(t)=\frac{\gamma(1-n(t))}{\overline{\beta(t)}}.
$$
A final form for the time $t$ coordinates of the fixed curve  in terms of $\overline{\beta(t)}$ is
\begin{equation}
\left(s(t),i(t),r(t)\right)=\left(\frac{\gamma^2}{\left(\overline{\beta(t)}-\lambda\right)^2},\frac{\lambda\gamma\left(\overline{\beta(t)}-\lambda-\gamma\right)}{\left[\overline{\beta(t)}\left(\overline{\beta(t)}-\lambda\right)\right]^2},\frac{\gamma\left(\overline{\beta(t)}-\lambda-\gamma\right)}{\overline{\beta(t)}\left(\overline{\beta(t)}-\lambda\right)}\right).\label{FPTraj}
\end{equation}

This curve is valid biologically only when $\overline{\beta(t)}-\lambda-\gamma>0$ (when both the trivial and disease free equilibria are saddle nodes).  Eigenvalues about this curve are $-\frac{\gamma\lambda}{\overline{\beta(t)}-\lambda}$ and
$$
\nonumber\frac{1}{2\overline{\beta(t)}\left(\overline{\beta(t)}-\lambda\right)^2}\left[b\pm\sqrt{b^2-4\overline{\beta(t)}^2\gamma\lambda\left(\overline{\beta(t)}-\lambda\right)^2\left(\overline{\beta(t)}-\gamma-\lambda\right)}\right],
$$
where $b=\lambda\left(\overline{\beta(t)}-\lambda\right)\left(\overline{\beta(t)}-\gamma-\lambda\right)\left(\overline{\beta(t)}-1\right)$.  When the endemic equilibrium trajectory is inside $T$ the sign of the real part of the two complicated eigenvalues is determined by the sign of $\overline{\beta(t)}-1$ (the simple eigenvalue is negative).  Thus, the real part of these eigenvalues are positive if and only if $\overline{\beta(t)}>1$. If the endemic equilibrium trajectory is not in $T$, $\overline{\beta(t)}<\gamma+\lambda$, then there is always a positive eigenvalue for the endemic equilibrium trajectory, causing it to be unstable.

Being able to show that $\lim_{t\rightarrow\infty}\overline{\beta(t)}=\epsilon\in[0,\overline{\beta(0)})$ and that $T$ is a proper bounding set for the dynamics results in a complete understanding of the transient and asymptotic dynamics of this system. Supposing $i_\infty=\lim_{t\rightarrow\infty}i(t)>0$ immediately gives that $n_\infty>0$ and for all $t$ the transport variable $v(t)$ is decreasing, ($\dot{v}(t)<0$ implying that $v(t)\rightarrow-\infty$).  To study the limiting behavior of $\overline{\beta(t)}$, suppose $s_\infty>0$.  Rewriting $\overline{\beta(t)}$ as
$$
\overline{\beta(t)}=\frac{\int\beta(w)s(0,w)e^{\beta(w)v(t)+u(t)+\lambda t}dw}{s(t)},
$$
one would need to consider $\lim_{t\rightarrow\infty}\left[\beta(w)v(t)+u(t)+\lambda t\right]\asymp-\infty+\infty$, an improper form. From this we may not determine what $\overline{\beta(t)}$ limits to other than (with the knowledge that it is monotonically decreasing) some value $\epsilon\in[0,\overline{\beta(0)})$. 

To show that $T$ is a proper bounding region we must show that on the boundary of $T$ all flow is inwards. The positivity conditions ($s(t),i(t)\geq0$) are straightforward, and thus we consider $s(t)+i(t)=1$.  Along this curve  $\dot{s}(t)$ is always negative:
$$
\nonumber\dot{s}(t)=-\overline{\beta(t)}s(t)(1-s(t))\leq0.
$$
The $i(t)$ differential equation simplifies to $(1-s(t))(\overline{\beta(t)}-\gamma)$, and thus $i(t)$ is increasing for $s(t)>\frac{\gamma}{\overline{\beta(t)}}$.  This could be problematic; if the magnitude of flow in the $i(t)$ direction is greater than that in the $s(t)$ direction then the flow would escape $T$.  Assuming that $n(t)=1$ and $s(t)\leq n(t)$, we have that $n(t)$ is decreasing.  Thus $s(t)$ is decreasing more than $i(t)$ increases and the flow remains within the region $T$! Therefore $T$ is a proper bounding region for the dynamics should $(s(0),i(0),n(0))\in T$. 

Several possible transient/asymptotic dynamic situations for the fixed points $$(s^*,i^*,r^*)=(0,0,0), (1,0,0), \left(\frac{\gamma^2}{\left(\epsilon-\lambda\right)^2},\frac{\lambda\gamma\left(\epsilon-\lambda-\gamma\right)}{\left[\epsilon(\epsilon-\lambda)\right]^2},\frac{\gamma\left(\epsilon-\lambda-\gamma\right)}{\epsilon\left(\epsilon-\lambda\right)}\right),$$ the trivial, disease free (DFE), and the limit of the endemic equilibrium (EE), are possible:
\begin{enumerate}
\item The limit $\epsilon$ is both greater than $\gamma+\lambda$ and $1$: The trivial equilibrium and the DFE are saddles.  The trivial equilibrium attracts along the $i(t)$ and $r(t)$ axes and repels along the $s(t)$ axis.  The DFE attracts along the $s(t)$ axis and the line $s(t)+r(t)=n(t)$ and repels in the direction of $s(t)+i(t)\leq n(t)$.  The endemic equilibrium (EE) is also a saddle with two eigenvalues with positive real part and one negative eigenvalue. Solutions will oscillate about the current value of \ref{FPTraj} and asymptotically approach oscillation about the limit of \ref{FPTraj}.\label{oscillations}

\item The limit $\epsilon$ is in $(\gamma+\lambda,1)$: The trivial equilibrium and the DFE are saddles. The EE is now an attractor for all trajectories in the interior of $T$.\label{EE}

\item The limit $\epsilon$ is less than $\gamma+\lambda$: The DFE is the attractor for all trajectories in the interior of $T$, and the EE is outside of $T$ repelling trajectories into $T$.\label{DFE}

\item $\overline{\beta(t)}>\gamma+\lambda$ for $t\in[0,\tau)$ and $\overline{\beta(t)}<\gamma+\lambda$ for $t\in(\tau,\infty)$: Until time $\tau$ the system will appear as in Case \ref{oscillations}, Case \ref{EE}, or Case \ref{oscillations} and then Case \ref{EE} (depending on the sign of $\overline{\beta(t)}-1$).  After time $\tau$ the behavior will be as in Case \ref{DFE}.\label{Both}
\end{enumerate}

\section{Discussion}
Albeit in a narrow context, this paper has focused on one way to handle biological heterogeneity, via transport equations akin to the reduction theory of Karev.  Lemma \ref{lemma: distributed model} follows from the monotonicity of $\overline{\beta(t)}$ which allows one to prove qualitative equivalence (not a new result, but a new way to show it) as well as demonstrating that the undistributed case infects the most individuals.  This latter result in closed populations is intuitive but was shown to hold for the pure inheritance model as well.  

The simple, open $SIR$ model was reduced to a non-autonomous, finite dimensional system of ODEs, but due to the resulting nature of $\overline{\beta(t)}$ (involving the solution trajectory of the transport variable) is challenging to analyze. Without the ability to demonstrate monotonicity (at least after some time $\tau$) we are not afforded with the ability to rule out, or construct conditions for, sustained oscillatory behavior.

The pure inheritance model, System \ref{opensystem}, is perhaps an unbiologically realistic work-around for the ``blue-sky" births found in the simple, open $SIR$ case, but it was in a form receptive to the transport system representation. The reduction of the homogeneous system to one with states in $[0,1]$ allowed for the fixed point analysis of the closed $SIR$ to be applicable to stability analysis of exponential trajectories. Interestingly, despite being an $SIR$ model with demographic dynamics, the births being split into the three classes prevents the presence of oscillatory solutions. This splitting of births is perhaps what confers monotonicity on $\overline{\beta(t)}$. 

The sterilization model, a simplification of papillomavirus dynamics in sheep given by System \ref{logisticopen}, exhibits a wide range of transient dynamics due to the stability of the DFE being dependent on the relationship between $\overline{\beta(t)}$ and the other vital rates $\lambda$ and $\gamma$ as well as its own magnitude relative to 1.  Motivation for splitting up the population into variable susceptibility is in the spirit of black boxing several cofactors: nutrition, cleanliness of environment, genetic variability, and epidemiological history of the individual sheep. Furthermore, the model does not incorporate the probable culling/removal of infected sheep, the partial infectivity of those sheep in the chronic phase. Furthermore, it was assumed that chronic infected (recovery) necessarily led to sterility when early detection and treatment can prevent the scarification from occurring although the sheep would probably be removed from the breeding population to prevent more infection. 

It was demonstrated that for particular values of the mean distribution of infectivity that the disease free equilibrium is asymptotically stable, Cases \ref{DFE} and \ref{Both}, and that in the latter situation there is an interesting transient behavior of the solution curve ``chasing" an equilibrium which vanishes from the biologically meaningful space.  I've also given conditions for oscillatory behavior should the mean susceptibility {\em not} limit to zero, Case \ref{oscillations}. In this situation the dynamics are quite complex; because there is not a fixed limit cycle for finite time; there is however, a ``moving" oscillation through the phase space (the trajectory is oscillating and where the oscillation occurs is moving).  Finally in Case \ref{EE}, I was able to show that the endemic state is a ``global" attractor within $T$.

Extensions may include a more general theoretic version of the transport system theory of Karev as applied to models requiring separation of variables to be ``solved", as in the case of the simple, open $SIR$. Additionally, heterogeneity may be introduced to infectivity and recovery rate as Novozhilov did for closed populations \citep{novozhilov2008spread}. These heterogeneities were not introduced here because they induce further transport equations, and was beyond the scope of an explanation of the method through novel examples. The question of parameter estimation was raised in Section \ref{ClosedSIR}. While the robustness of estimators of parameters in epidemiological compartmental models has been demonstrated, \citep{Geoffard95}, for all but behavioral effects, it is of interest to investigate if the distribution, or the parameters of an assumed distribution, of an epidemiological parameter may be estimated from collected data. 

\section{Acknowledgements}
I would like to thank Georgy Karev and Artem Novozhilov for helping to answer my technical questions about the method. Additionally I acknowledge Juan Aparicio for putting in plain language some of the finer points.

\hspace{12 mm}
\bibliography{Paperbib}

\end{document}